\documentclass[runningheads]{llncs}
\usepackage{amsmath,amssymb}
\usepackage{graphicx,rotating}
\usepackage{paralist,soul,booktabs}
\usepackage{todonotes,hyperref,xspace}
\usepackage[capitalise]{cleveref}
\usepackage{subcaption}
\usepackage{lineno}
\usepackage[misc]{ifsym}

\usepackage{xspace}
\usepackage{todonotes}
\usepackage{thm-restate}

\renewcommand{\qed}{\hfill $\square$}

\makeatletter
\DeclareRobustCommand{\bfseries}{%
  \not@math@alphabet\bfseries\mathbf
  \fontseries\bfdefault\selectfont
  \boldmath
}
\makeatother

\Crefname{figure}{Fig.}{Figs.}

\title{Sketched Representations and Orthogonal Planarity of Bounded Treewidth Graphs\thanks{
Research partially supported by: (i) MIUR, grant 20174LF3T8 AHeAD: efficient Algorithms for HArnessing networked Data.; (ii) Engineering Dep. - University of Perugia,  grants RICBASE2017WD and RICBA18WD: ``Algoritmi e sistemi di analisi visuale di  reti complesse e di grandi dimensioni''.
}}

\author{
Emilio Di Giacomo\orcidID{0000-0002-9794-1928}\and
Giuseppe Liotta\orcidID{0000-0002-2886-9694}\and
Fabrizio Montecchiani\orcidID{0000-0002-0543-8912} \Letter
}

\institute{Department of Engineering, University of Perugia, Italy\\
	\email{\{emilio.digiacomo, giuseppe.liotta, fabrizio.montecchiani\}@unipg.it}
}

\authorrunning{E. Di Giacomo \and G. Liotta \and F. Montecchiani}

\titlerunning{Orthogonal Planarity of Bounded Treewidth Graphs}

\newcommand{\p}{\pi}
\newcommand{\ph}{\frac{\pi}{2}}
\newcommand{\pt}{\frac{3\pi}{2}}
\newcommand{\compact}{sketched embedding\xspace}
\newcommand{\shaped}{orthogonal sketch\xspace}
\newcommand{\algo}{\textsc{OrthoPlanTester}\xspace}
\newcommand{\hv}{\textsc{HV-Planarity}\xspace}
\newcommand{\flex}{\textsc{FlexDraw}\xspace}
\newcommand{\ortho}{\textsc{OrthogonalPlanarity}\xspace}
\newenvironment{sketch}
{\noindent{\bf {Proof sketch.}}}{\hspace*{\fill}\qed\par\vspace{2mm}}

\begin{document}

\maketitle


\begin{abstract}
Given a planar graph $G$ and an integer $b$, \ortho is the problem of deciding whether $G$ admits an orthogonal drawing with at most $b$ bends in total. We show that \ortho can be solved in polynomial time if $G$ has bounded treewidth. Our proof is based on an FPT algorithm whose parameters are the number of bends, the treewidth and the number of degree-2 vertices of $G$. This result is based on the concept of sketched orthogonal representation that synthetically describes a family of equivalent orthogonal representations. Our approach can be extended to related problems such as \hv and \flex. In particular,  both \ortho and \hv can be decided in $O(n^3 \log n)$ time for series-parallel graphs, which improves over the previously known $O(n^4)$ bounds.
\end{abstract}

\section{Introduction}

An \emph{orthogonal drawing} of a planar graph $G$ is a planar drawing where each edge is drawn as a chain of horizontal and vertical segments; see~\cref{fi:compact-a}. Orthogonal drawings are among the most investigated research subjects in graph drawing, see, e.g.,~\cite{DBLP:journals/comgeo/BiedlK98,DBLP:journals/talg/BiedlLPS13,DBLP:journals/algorithmica/BlasiusKRW14,DBLP:conf/compgeom/ChangY17,DBLP:journals/siamcomp/BattistaLV98,DBLP:conf/gd/DidimoLP18,DBLP:journals/siamcomp/GargT01,DBLP:journals/ieicet/RahmanEN05,t-eggmnb-87,DBLP:journals/siamdm/ZhouN08} for a limited list of references, and also~\cite{dett-gdavg-99,DBLP:reference/crc/DuncanG13} for surveys.
The \ortho problem asks whether $G$ admits an orthogonal drawing with at most $b$ bends in total, for a given $b \in \mathbb{N}$.

In a seminal paper, Garg and Tamassia~\cite{DBLP:journals/siamcomp/GargT01} proved that \ortho is NP-complete when $b=0$, which implies that minimizing the number of bends is also NP-hard. In fact, it is even NP-hard to approximate the minimum number of bends in an orthogonal drawing with an $O(n^{1-\varepsilon})$ error for any $\varepsilon > 0$~\cite{DBLP:journals/siamcomp/GargT01}. On the positive side, Tamassia~\cite{t-eggmnb-87} showed that \ortho can be decided in polynomial time if the input graph is \emph{plane}, i.e., it has a fixed embedding in the plane. When a planar embedding is not given as part of the input, polynomial-time algorithms exist for some restricted cases, namely subcubic planar graphs and series-parallel graphs, see, e.g.,~\cite{DBLP:conf/compgeom/ChangY17,DBLP:journals/siamcomp/BattistaLV98,DBLP:conf/gd/DidimoLP18,DBLP:journals/ieicet/RahmanEN05,DBLP:journals/siamdm/ZhouN08}.

Given the hardness result for \ortho, a natural research direction is to investigate its parameterized complexity~\cite{DBLP:series/mcs/DowneyF99}. Despite the rich literature about orthogonal drawings, this direction has been surprisingly disregarded. The only exception is a result by Didimo and Liotta~\cite{DBLP:conf/isaac/DidimoL98}, who described an algorithm for  biconnected planar graphs that runs in $O(6^{r} n^4 \log n)$ time, where $r$ is the number of degree-4 vertices.
We recall that FPT algorithms have been proposed for other graph drawing problems, such as upward planarity~\cite{DBLP:conf/esa/Chan04,DBLP:journals/siamdm/DidimoGL09,DBLP:journals/ijfcs/HealyL06}, layered drawings~\cite{DBLP:journals/algorithmica/DujmovicFKLMNRRWW08}, linear layouts~\cite{DBLP:journals/jgaa/BannisterE18,DBLP:journals/jda/DujmovicFK08,DBLP:journals/algorithmica/DujmovicW04}, and $1$-planarity~\cite{DBLP:journals/jgaa/BannisterCE18}.

\smallskip\noindent\textbf{Contribution.} We describe an FPT algorithm for \ortho whose parameters are the number of bends, the treewidth and the number of degree-2 vertices of the input graph.
We recall that the notion of treewidth~\cite{DBLP:journals/jal/RobertsonS86} is commonly used as a parameter in the parameterized complexity analysis (see also \cref{se:preliminaries}). The algorithm works for planar graphs of degree four with no restriction on the connectivity. Our main contribution is summarized as follows.

\begin{theorem}\label{thm:main}
Let $G$ be an $n$-vertex planar graph with $\sigma$ degree-$2$ vertices and let $b \in \mathbb{N}$.
Given a tree-decomposition of $G$ of width $k$, there is an algorithm that decides \ortho in $f(k,\sigma, b) \cdot n$ time, where $f(k,\sigma, b)=k^{O(k)}(\sigma+b)^{k}\log (\sigma+b)$. The algorithm computes a drawing of $G$, if one exists.
\end{theorem}

\noindent For an $n$-vertex graph $G$ of treewidth $k$, a tree-decomposition of $G$ of width $k$ can be found in $k^{O(k^3)} \, n$ time~\cite{DBLP:journals/siamcomp/Bodlaender96}, while a tree-decomposition of width $O(k)$ can be computed in $2^{O(k)} \, n$ time~\cite{DBLP:journals/siamcomp/BodlaenderDDFLP16}.
The function $f(k,\sigma,b)$ depends exponentially on neither $\sigma$ nor $b$.
Since both $\sigma$ and $b$ are $O(n)$~\cite{DBLP:journals/comgeo/BiedlK98}, \ortho can be solved in time $n^{g(k)}$ for some polynomial function $g(k)$, and thus it belongs to the XP class when parameterized by treewidth~\cite{DBLP:series/mcs/DowneyF99}. Moreover, since the number of bends in a bend-minimum orthogonal drawing is $O(n)$, the next result follows from \cref{thm:main}, performing a binary search on $b$.

\begin{corollary}\label{co:rpt-bt}
Let $G$ be an $n$-vertex planar graph.
Given a tree-decomposition of $G$ of width $k$, there is an algorithm that decides \ortho in $k^{O(k)}n^{k+1}\log n$ time. Also, a bend-minimum orthogonal drawing of $G$ can be computed in $k^{O(k)}n^{k+1}\log^2 n$ time.
\end{corollary}

\noindent By \cref{co:rpt-bt} \ortho can be decided in $O(n^{3} \log n)$ time for graphs of treewidth two, and hence bend-minimum orthogonal drawings can be computed in $O(n^{3} \log^2 n)$ time. We remark that the best previous result for these graph, dating back to twenty years ago, is an $O(n^4)$ algorithm by Di Battista et al.~\cite{DBLP:journals/siamcomp/BattistaLV98} which however is restricted to biconnected graphs (whereas ours is not).


%
%
%
%

%

\smallskip

Our FPT approach can be applied to related problems, namely to \hv and \flex. \hv takes as input a planar graph $G$ whose edges are each labeled H (horizontal) or V (vertical) and it asks whether $G$ admits an orthogonal drawing with no bends and in which the direction of each edge is consistent with its label. As a corollary of our results, we can decide \hv in  $O(n^3\log n)$ time for series-parallel graphs, which improves a recent $O(n^4)$ bound by Didimo et al.~\cite{DBLP:journals/jcss/DidimoLP19} and addresses one of the open problems in that paper. \flex takes as input a planar graph $G$ whose edges have integer weights and it asks whether $G$ admits an orthogonal drawing where each edge has a number of bends that is at most its weight~\cite{DBLP:journals/algorithmica/BlasiusKRW14,DBLP:journals/comgeo/BlasiusLR16}.

\smallskip\noindent\textbf{Proof strategy and paper organization.} The first ingredient of our approach is a well-known combinatorial characterization of orthogonal drawings (see~\cite{dett-gdavg-99,t-eggmnb-87}) that transforms \ortho to the problem of testing the existence of a planar embedding along with a suitable angle assignment to each vertex-face and edge-face incidence (see \cref{se:preliminaries}). The second ingredient is the definition of a suitable data structure, called \emph{\shaped{es}}, that encodes sufficient information about any such combinatorial representation, and in particular it makes it possible to decide whether the representation can be extended with further vertices incident to a given vertex cutset of the graph (see \cref{se:compact}). The proposed algorithm (see \cref{se:algo}) traverses a tree-decomposition of the input graph and stores a limited number of \shaped{es} for each node of the tree, rather than all its possible orthogonal drawings. This number depends on the width of the tree-decomposition, on the number of bends, and on the number of degree-2 vertices. The key observation is that a vertex of degree greater than two may correspond to a right turn when walking clockwise along the boundary of a face but not to a left turn, while a degree-2 vertex  may correspond to both a left or a right turn. Thus, the number of degree-2 vertices, as well as the number of bends, have an impact in how much a face can ``roll-up'' in the drawing, which in our approach translates in possible weights that can be assigned to the edges of an \shaped. The extensions of our approach can be found in \cref{se:extensions}, while conclusions and open problems are in \cref{se:conclusions}. For reasons of space, some proofs have been moved to the appendix and the corresponding statements are marked with an asterisk~(*).

\section{Preliminaries}\label{se:preliminaries}

\smallskip\noindent\textbf{Embeddings.}
%
We assume familiarity with basic notions about graph drawings.
A planar drawing of a planar graph $G$ subdivides the plane into topologically connected regions, called \emph{faces}. The infinite region is the \emph{outer face}. A \emph{planar embedding} of $G$ is an equivalence class of planar drawings that define the same set of faces and with the same outer face. A planar embedding of a connected graph can be uniquely identified by specifying its \emph{rotation system}, i.e., the clockwise circular order of the edges around each vertex, and the outer face. A \emph{plane graph} $G$ is a planar graph with a given planar embedding. The number of vertices encountered in a closed walk along the boundary of a face $f$ of $G$ is the \emph{degree} of $f$, denoted as $\delta(f)$. If $G$ is not $2$-connected, a vertex may be encountered more than once, thus contributing  more than one unit to the degree of the face.

\smallskip\noindent\textbf{Orthogonal Representations.}
Let $G=(V,E)$ be a planar graph with vertex degree at most four. A planar drawing $\Gamma$ of $G$ is \emph{orthogonal} if each edge is a polygonal chain consisting of horizontal and vertical segments. A \emph{bend} of an edge $e$ in $\Gamma$ is a point shared by two consecutive segments of $e$. An angle formed by two consecutive segments incident to the same vertex (resp. bend) is a \emph{vertex-angle} (resp. \emph{bend-angle}). An orthogonal representation of $G$ can be derived from $\Gamma$ and it  specifies the values of all vertex- and bend-angles (see~\cite{dett-gdavg-99,t-eggmnb-87}). 
More formally, let $\mathcal E$ be a planar embedding of $G$, and let $e=(u,v)$ be an edge that belongs to the boundary of a face $f$ of $\mathcal E$. The two possible orientations $(u,v)$ and $(v,u)$ of $e$ are called \emph{darts}. A dart is \emph{counterclockwise with respect to $f$}, if $f$ is on the left side when walking along the dart following its orientation. Let $D(u)$ be the set of darts exiting from $u$ and let $D(f)$ be the set of counterclockwise darts~with~respect~to~$f$.

\begin{definition}\label{def:ortho}
Let $G=(V,E)$ be a planar graph with vertex degree at most four.
An \emph{orthogonal representation} $H$ of $G$ is a planar embedding $\mathcal E$ of $G$ and an assignment to each dart $(u,v)$ of two values $\alpha(u,v) = c_{\alpha} \cdot \ph$ and $\beta(u,v) = c_{\beta} \cdot \ph$, where $c_{\alpha} \in \{1,2,3,4\}$ and $c_{\beta} \in \mathbb{N}$, that satisfies the following conditions.
\begin{description}
\item[\bf C1.] For each vertex $u$: $\sum\limits_{(u,v) \in D(u)}\alpha(u,v)=2\p$;
\item[\bf C2.] For each internal face $f$: $\sum\limits_{(u,v) \in D(f)}(\alpha(u,v)+\beta(v,u)-\beta(u,v))=\p(\delta(f)-2)$;
\item[\bf C3.] For the outer face $f_{o}$: $\sum\limits_{(u,v) \in D(f_{o})}(\alpha(u,v)+\beta(v,u)-\beta(u,v))=\p(\delta(f_{o})+2)$.
\end{description}
\end{definition}

\noindent Let $f$ be the face counterclockwise with respect to dart $(u,v)$. The value $\alpha(u,v)$ represents the vertex-angle that dart $(u,v)$ forms with the dart following it in the circular counterclockwise order around $u$; we say that $\alpha(u,v)$ is a \emph{vertex-angle of $u$ in $f$}. The value $\beta(u,v)$ represents the sum of the $\ph$ bend-angles that dart $(u,v)$ forms in $f$. Condition {\bf C1} guarantees that the sum of angles around each vertex is valid, while  {\bf C2} (respectively, {\bf C3}) guarantees that the sum of the angles at the vertices  and at the bends of an internal face (respectively, outer face) is also valid. 
Given an orthogonal representation of an $n$-vertex graph $G$, a corresponding orthogonal drawing can be computed in $O(n)$ time~\cite{t-eggmnb-87}.

\smallskip\noindent\textbf{Tree-decompositions.}
Let $(\mathcal{X},T)$ be a pair such that $\mathcal{X}=\{X_1,X_2,\dots,X_\ell\}$ is a collection of subsets of vertices of a graph $G$ called \emph{bags}, and $T$ is a tree whose nodes are in a one-to-one mapping with the elements of $\mathcal X$. With a slight abuse of notation, $X_i$ will denote both a bag of $\mathcal{X}$ and the node of $T$ whose corresponding bag is $X_i$. The pair $(\mathcal{X},T)$ is a \emph{tree-decomposition} of $G$ if
: (i) For every edge $(u,v)$ of $G$, there exists a bag $X_i$ that contains both $u$ and $v$, and (ii) For every vertex $v$ of $G$, the set of nodes of $T$ whose bags contain $v$ induces a non-empty (connected) subtree of $T$.
%
%
%
%
%
The \emph{width} of a tree-decomposition $(\mathcal{X},T)$ of $G$ is $\max_{i=1}^\ell {|X_i| - 1}$, and the \emph{treewidth} of $G$ is the minimum width of any tree-decomposition of $G$.
We  use a particular tree-decomposition (which always exists~\cite{DBLP:books/sp/Kloks94}) that limits the number of possible transitions between bags. 

\begin{definition}\emph{\cite{DBLP:books/sp/Kloks94}}\label{def:nice}
 A tree-decomposition $(\mathcal{X},T)$ of $G$ is \emph{nice} if $T$ is a rooted tree and: (a) Every node of $T$ has at most two children, (b) If a node  $X_i$ of $T$ has two children whose bags are $X_j$ and $X_{j'}$, then $X_i=X_j=X_{j'}$, (c) If a node $X_i$ of $T$ has only one child $X_j$, then there exists a vertex $v \in G$ such that either $X_i = X_j \cup \{v\}$ or $X_i \cup \{v\} = X_j$. In the former case of (c) we say that $X_i$ \emph{introduces} $v$, while in the latter case  $X_i$ \emph{forgets} $v$.
\end{definition}

\section{Orthogonal Sketches}\label{se:compact}

Recall that an orthogonal  representation of a planar graph $G$ corresponds to a planar embedding of $G$ and to an assignment of vertex- and bend-angles in each face of $G$. A fundamental observation for our approach is that the conditions that make an assignment of such angles a valid orthogonal representation of $G$ can be verified for each vertex and for each face independently. 
In what follows we define two equivalence relations on the set of orthogonal representations of $G$ that yields a set of equivalence classes whose size is bounded by some function of the width of $T$, of the number of degree-$2$ vertices of $G$, and of the~number~of~bends.


\begin{figure}[t]
	\centering
	\begin{minipage}[b]{.48\textwidth}
		\centering
		\includegraphics[page=1, width=\textwidth]{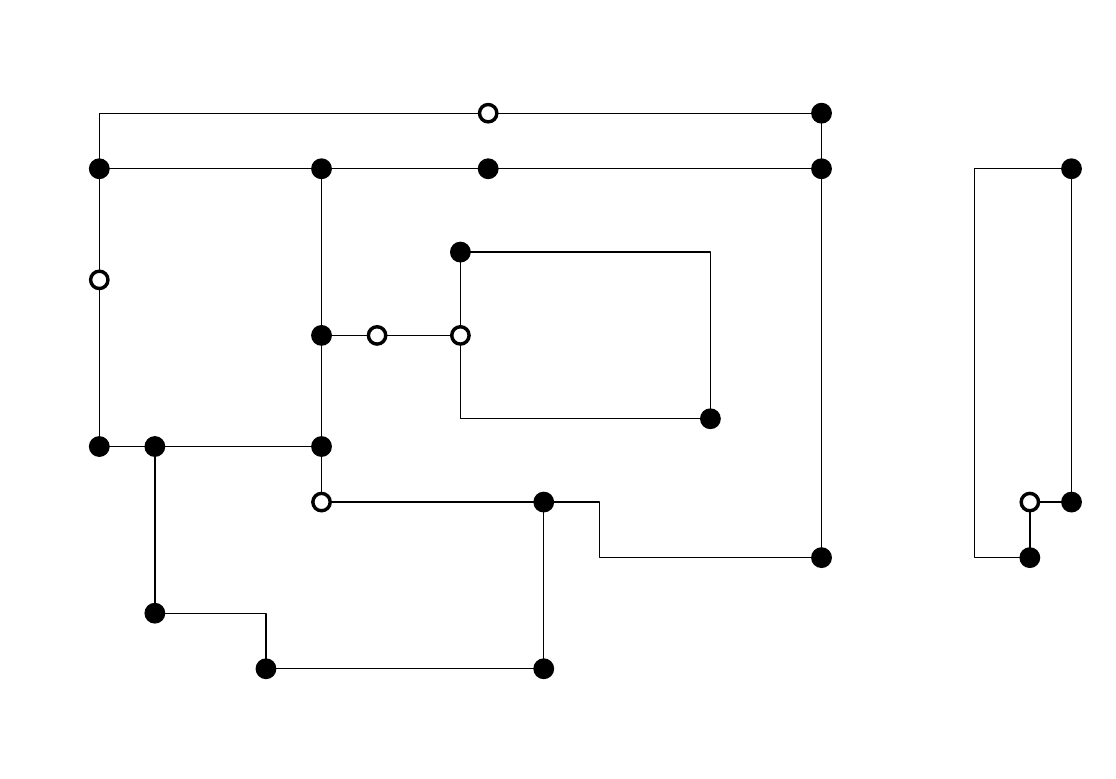}
		\subcaption{$\Gamma$}\label{fi:compact-a}
	\end{minipage}
	\hfil
	\begin{minipage}[b]{.48\textwidth}
		\centering
		\includegraphics[page=2, width=\textwidth]{figures/compact-example}
		\subcaption{~}\label{fi:compact-b}
	\end{minipage}
	\hfil
	\begin{minipage}[b]{.48\textwidth}
		\centering
		\includegraphics[page=4, width=\textwidth]{figures/compact-example}
		\subcaption{$C^*(H,X)$}\label{fi:compact-d}
	\end{minipage}
	\hfil
	\begin{minipage}[b]{.48\textwidth}
		\centering
		\includegraphics[page=3, width=\textwidth]{figures/compact-example}
		\subcaption{$\langle C(H,X), \phi,\rho\rangle$}\label{fi:compact-c}
	\end{minipage}
	\caption{\label{fi:compact-example} (a) An orthogonal drawing $\Gamma$  of a graph $G=(V,E)$ with $8$ bends; the white vertices define a set $X \subseteq V$. (b) The representing cycles of the active faces of $H$ with respect to $X$, where $H$ denotes the orthogonal representation of $\Gamma$. (c) The connected \compact $C^*(H,G)$. (d) The \shaped $\langle C(H,X), \phi,\rho\rangle$.}
\end{figure}

\smallskip\noindent\textbf{Sketched Embeddings.}
Let $H$ be an orthogonal representation of a planar graph $G=(V,E)$ and let $X \subseteq V$; see for example \cref{fi:compact-a}.
The vertices in $X$ are called \emph{active}. A face $f$ of $H$ is \emph{active} if it contains at least one active vertex. 
A \emph{representing cycle} $C_f$ of an active face $f$ is an oriented cycle such that: (i) It contains all and only the active vertices of $f$ in the order they appear in a closed walk along the boundary of $f$. (ii) $C_f$ is counterclockwise with respect to $f$, that is, $C_f$ is oriented coherently with the counterclockwise darts of face $f$. Notice that $C_f$ may be non-simple because a cut-vertex may appear multiple times when walking along $C_f$. Also, if $H$ contains distinct components, the outer face of each component is considered independently.  See \cref{fi:compact-b} for an illustration.

Let $H$ be an orthogonal representation of a planar graph $G$. We may conveniently focus on an orthogonal drawing $\Gamma$ that falls in the equivalence class of drawings having $H$ as an orthogonal representation. Assume first that $H$ is connected. The \emph{\compact} of $H$ with respect to $X$ is the plane graph $C(H,X)$ constructed as follows. For each active face $f$ we draw in $\Gamma$ its representing cycle $C_f$ by identifying the vertices of $C_f$ with the corresponding vertices of $f$ and by drawing the edges of $C_f$ inside $f$ without creating crossings. Graph $C(H,X)$ is the embedded graph formed by the edges that we drew inside the active faces. This is a plane graph by construction, it may be disconnected, and it may contain self-loops and multiple edges. Graph $C(H,X)$ has a face $f'$ for each representing cycle $C_f$ of an active face $f$ of $H$; we call $f'$ an \emph{active} face of $C(H,X)$.
If $H$ is not connected, a \compact $C(H_i,X)$ is computed for each connected component $H_i$ of $H$ ($i=1,\dots,h$) and the \compact of $H$ is $C(H,X)= \bigcup_{i=1}^{h} C(H_i,X)$. See for example \cref{fi:compact-b}.
 

\begin{definition}\label{def:x-equiv}
	Let $G=(V,E)$ be a planar graph and let $X \subseteq V$.
	Let $H_1$ and $H_2$ be two orthogonal representations of $G$.
	$H_1$ and $H_2$ are \emph{$X$-equivalent} if they have the same \compact.
\end{definition}

Suppose that $H$ is connected. We now aim at computing a connected supergraph of $C(H,X)$.
By construction, the active faces of $C(H,X)$ may share vertices but not edges and hence $C(H,X)$ also contains faces that are not active. For each non-active face $g$ of $C(H,X)$, we add a dummy vertex $v_g$ in its interior and we connect it to all vertices on the boundary of $g$ by adding dummy edges. This turns $C(H,X)$ to a connected plane graph $C^*(H,X)$, which we call a \emph{connected} \compact of $H$ with respect to $X$. If $H$ is not connected, a connected \compact $C^*(H_i,X)$ is computed for each $C(H_i,X)$ independently, and the \emph{connected} \compact of $H$ is $C^*(H,X) = \bigcup_{i=1}^{h} C^*(H_i,X)$. \cref{fi:compact-d} shows a connected \compact obtained from \cref{fi:compact-b}. Observe that it may be possible to construct different connected \compact{s} of the same \compact. However, any connected \compact encodes the  information about the global structure of $H$ that is sufficient for the purposes of our algorithm. 
$C^*(H,X)$ (and hence $C(H,X)$) has a number of vertices that is $O(|X|)$ and a number of edges that is also $O(|X|)$ because $C^*(H,X)$ is planar and the multiplicity of an edge in $C^*(H,X)$ is at most four. 

\begin{restatable}[*]{lemma}{numofcompact}
	\label{le:num-of-compact}
	Let $G=(V,E)$ be a planar graph  and let $X \subseteq V$.
	Let $\mathcal{H}$ be the set of all possible orthogonal representations of $G$.
	The $X$-equivalent relation partitions $\mathcal{H}$ in at most $w^{O(w)}$ equivalence classes, where $w=|X|$.
\end{restatable}

\smallskip\noindent\textbf{Orthogonal Sketches.}
Let $H$ be an orthogonal representation of a plane graph $G=(V,E)$ and let $X \subseteq V$. Let $C(H,X)$ be a \compact of $H$ with respect to $X$. Recall that $H$ is defined by two functions, $\alpha$ and $\beta$, that assign the vertex- and bend-angles made by darts inside their faces. The \emph{shape} of $C(H,X)$ consists of two functions $\phi$ and $\rho$ defined as follows. Let $(u,v)$ be a dart of $C(H,X)$, which corresponds to a path $\Pi_{uv}$ in $H$. Let $z$ be the vertex of $\Pi_{uv}$ adjacent to $u$ (possibly $z=v$). We set $\phi(u,v) = \alpha(u,z)$; the value $\phi(u,v)$ still represents the vertex-angle that $u$ makes in the face on the left of $(u,v)$. Function $\rho$ assigns to each dart $(u,v)$ of $C(H,X)$ a number that describes the shape of $\Pi_{uv}$ in $H$. More precisely, for each representing cycle $C_f$ and for each dart $(u,v)$ of $C_f$, $\rho(u,v,f)=n_\ph(u,v)-n_\pt(u,v)-2n_{2\p}(u,v)$, where $n_a(u,v)$ ($a \in \{\ph,\pt,2\p\}$) is the number of vertex- and bend-angles between $u$ and $v$ whose value is $a$. For example, \cref{fi:compact-c} shows a \compact together with its shape.  We call $\rho(u,v,f)$  the \emph{roll-up number}\footnote{It may be worth observing that other papers used conceptually similar definitions, called \emph{rotation} (see, e.g.,~\cite{DBLP:journals/algorithmica/BlasiusKRW14}) and \emph{spirality} (see, e.g.,~\cite{DBLP:journals/siamcomp/BattistaLV98}).} of $(u,v)$ in $f$. If $\phi(u,v)>\ph$ and $f$ is the counterclockwise face with respect to dart $(u,v)$, we say that $u$ is \emph{attachable} in $f$.  Two $X$-equivalent \compact{s} have the \emph{same shape}, if they have the same values of $\phi$ and $\rho$. A \compact $C(H,X)$, together with its shape $\langle \phi, \rho \rangle$, is called an \emph{\shaped} and it is denoted by $\langle C(H,X), \phi, \rho \rangle$.

\begin{definition}\label{def:shape-equiv}
	Let $G=(V,E)$ be a planar graph and let $X \subseteq V$.
	Let $H_1$ and $H_2$ be two  orthogonal representations of $G$.
	$H_1$ and $H_2$ are \emph{shape-equivalent} if they are $X$-equivalent and their \shaped{es} have the same shape.
\end{definition}


\begin{restatable}[*]{lemma}{validshape}
	\label{le:valid-shape}
	Let $\langle C(H,X), \phi, \rho \rangle$ be an \shaped.
	Let $C_f$ be a representing cycle of $C(H,X)$ and consider a closed walk along its boundary.
	Let $\rho^*$ be the sum of the roll-up numbers over all the traversed edges, and let $n_a$ be the number of encountered vertex-angles with value $a \in \{\ph, \pt, 2\p\}$. Then $\rho^*+n_\ph-n_\pt-2n_{2\p}=c$, with $c=4$ ($c=-4$) if $f$ is an inner (the outer) face.
\end{restatable}

\begin{restatable}[*]{lemma}{numofshapedcompact}
	\label{le:num-of-shaped-compact}
	Let $G=(V,E)$ be a planar graph with $\sigma$ vertices of degree two. Let $\mathcal{H}$ be the set of all possible orthogonal representations of $G$ with at most $b$ bends in total. The shape-equivalent relation partitions $\mathcal{H}$ in at most $w^{O(w)} \cdot (\sigma+b)^{w-1}$ equivalence classes, where $w=|X|$.
\end{restatable}
\begin{sketch}
We shall prove that $n_X \cdot n_S \leq w^{O(w)} \cdot (\sigma+b)^{w-1}$, where $n_X$ is the number of $X$-equivalent classes and $n_S$ is the number of possible shapes for each $X$-equivalent class. By \cref{le:num-of-compact}, $n_X \leq w^{O(w)}$; we can show  that $n_S \leq w^{O(w)}(\sigma+b)^{w-1}$. For a fixed \compact $C(H,X)$, a shape is defined by assigning to each dart $(u,v)$ of $C(H,X)$ the two values $\phi(u,v)$ and $\rho(u,v,f)$. The number of choices for the values $\phi(u,v)$ is at most $4^{4w} \leq w^{O(w)}$. As for the possible choices for $\rho(u,v,f)$, we claim that $-(\sigma+b) \le \rho(u,v,f) \le \sigma+b+4$ based on two observations. (1) The number of vertices and bends forming an angle of $\pt$ inside a face cannot be greater than $b + \sigma$. 	(2) For each vertex forming an angle of $2\p$ inside a face there are two vertex-angles of $\ph$ inside the same face. Finally, once the vertex-angles are fixed, the number of darts for which the roll-up number can be fixed independently is at most $w-1$. Thus, we have $w-1$ values to choose and $2(\sigma+b)+5$ choices for each of them.~\end{sketch}

\section{The Parameterized Algorithm}\label{se:algo}

\textbf{Overview.} We describe an algorithm, called \algo, that decides whether a planar graph $G$ admits an orthogonal drawing with at most $b$ bends in total, by using a dynamic programming approach on a nice tree-decomposition $T$ of $G$. The algorithm traverses $T$ bottom-up and decides whether the subgraph associated with each subtree admits an orthogonal drawing with at most $b$ bends. For each bag $X$, it stores all possible \shaped{es} and, for each of them, the minimum number of bends of any orthogonal representation encoded by that \shaped. To generate this record, \algo executes one of three possible procedures based on the type of transition with respect to the children of $X$ in $T$. If the execution of the procedure results in at least one \shaped, then the algorithm proceeds, otherwise it halts and returns a negative answer. If the root bag contains at least one \shaped, then the algorithm returns a positive answer. In the positive case, the information corresponding to the embedding of the graph  and to the vertex- and bend-angles can be reconstructed through a top-down traversal of $T$ so to obtain an orthogonal representation of $G$, and consequently an orthogonal drawing~\cite{t-eggmnb-87}.

\smallskip\noindent\textbf{The algorithm.} Let $G$ be an $n$-vertex planar graph with vertex degree at most four and with $\sigma$ vertices of degree two, and let $(\mathcal{X},T)$ be a nice tree-decomposition of $G$ of width $k$. Following a bottom-up  traversal of $T$, let $X_i$ be the next bag to be visited and let $B_i$ the set of all \shaped{es} of $X_i$. Let $w=k+1$ and recall that $|X_i| \le w$. Let $T_i$ be the subtree of $T$ rooted at $X_i$. Let $G_i$ be the subgraph of $G$ induced by all the vertices that belong to the bags in $T_i$. We distinguish the following four cases.

\textbf{$X_i$ is a leaf.} Without loss of generality, we can assume that $X_i$ contains only one vertex $v$ (as otherwise we can root in $X_i$ a chain of bags that introduce the vertices of $X_i$ one by one). Thus, $G_i$ contains only $v$ and it admits exactly one orthogonal representation with no bends. In particular, there is a unique \compact consisting of a single representing cycle $C_f$ having $v$ and no edges on its boundary. Also, the functions $\phi$ and $\rho$  are undefined.

\textbf{$X_i$ forgets a vertex.} Let $v$ be the vertex forgotten by $X_i$. Let $X_j$ be the child of $X_i$ in $T$.  In this case $G_i = G_j$ and we generate the \shaped{es} for $B_i$ by suitably updating those in $B_j$.
For each \shaped $\langle C(H,X_j), \phi, \rho \rangle$ in $B_j$ and for each representing cycle $C_f$ of $C(H,X_j)$ containing $v$, we apply the following operation. If $v$ is the only vertex of $C_f$, we remove $C_f$ from $C(H,X)$. Otherwise there are at most eight edges of $C_f$ incident to $v$, based on whether $v$ appears one or more times in a closed walk along $C_f$. We first remove all self-loops incident to $v$, if any. Let $(u_1,v)$, $(v,u_2)$ be any two edges of $C_f$ incident to $v$ that appear consecutively in a counterclockwise walk along $C_f$. For any such pair of edges we apply the following procedure. We remove the edges $(u_1,v)$, $(v,u_2)$ from $C_f$ and we add an edge $(u_1,u_2)$. We assign to the dart $(u_1,u_2)$ roll-up number equal to the sum of the roll-up numbers of darts $(u_1,v)$ and $(v,u_2)$ plus a constant $c$ defined as follows. If $\phi(v,u_2)=\p$, then $c=0$; if $\phi(v,u_2)=\ph$, then $c=1$; if $\phi(v,u_2)=\pt$, then $c=-1$; if $\phi(v,u_2)=2\p$, then $c=-2$. Once all consecutive pairs of edges incident to $v$ have been processed, $v$ is removed from $C_f$.
%
It is immediate to verify that \cref{le:valid-shape} holds for $C_f$ after applying this operation. See \cref{fi:introduce-c,fi:introduce-d} for an illustration.

	\begin{figure}[t]
		\centering
		\begin{minipage}[b]{.4\textwidth}
			\centering
			\includegraphics[page=3, width=\textwidth]{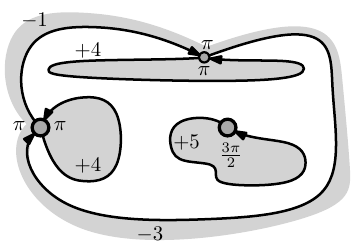}
			\subcaption{~}\label{fi:introduce-c}
		\end{minipage}
		\hfil
		\begin{minipage}[b]{.4\textwidth}
			\centering
			\includegraphics[page=4, width=\textwidth]{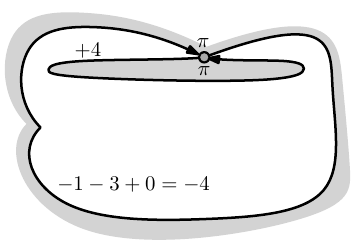}
			\subcaption{~}\label{fi:introduce-d}
		\end{minipage}
		\caption{ A portion of a \shaped before and after removing the bigger vertices. }
	\end{figure}

The above operation does not change the number of bends associated with the resulting \shaped{es}, but it may create duplicated \shaped{es} for $B_i$, which we delete. When deleting the duplicates, we shall pay attention on pairs of \shaped{es} that are the same but with a different number of bends. To see this, let $(u,v)$ be a dart of an \shaped $\langle C(H,X_j), \phi, \rho \rangle$, which corresponds to a path $\Pi_{uv}$ in $H$, and let $f$ be the face on the left of this path in $H$. An angle along $\Pi_{uv}$ in $f$  may be both a vertex-angle or a bend-angle. Hence, removing $v$ from different \shaped{es} (with different numbers of bends) may result in a set of \shaped{es} that differ only in the number of bends (see also \cref{fi:removesameshape} for an illustration). In this case, the algorithm stores the one with fewer bends, because in every step of the algorithm (see also the next two cases), the information about the total number of bends of an \shaped is only used to verify that it does not exceed the given parameter $b$.

\textbf{$X_i$ introduces a vertex.} Let $v$ be the vertex introduced in $X_i$. Let $X_j$ be the child of $X_i$ in $T$.
If $v$ does not have neighbors in $X_i$, then $B_i$ is the union of $B_j$ and the (unique) \shaped of the graph with the single vertex $v$ (see the leaf case). Otherwise, let $u_1,\dots,u_h$ be the neighbors of $v$ in $X_i$, with $h \le w$. We generate $B_i$ from $B_j$ by applying the following procedure. At a high level, we first update each \compact that can be extracted from an \shaped in $B_j$ by adding $v$, we then generate all shapes for the resulting \compact{s}, and we finally discard those shapes that are not valid. 

Let $C(H,X_j)$ be a \compact for which there is at least one \shaped in $B_j$. Suppose first that $u_1,\dots,u_h$ all belong to the same component of $C(H,X_j)$. By planarity, an orthogonal representation of $G_j$, whose \compact is $C(H,X_j)$, can be extended with $v$ only if it contains at least one face having all of $v$'s neighbors on its boundary. This corresponds to verifying first the existence of a representing cycle $C_f$ in $C(H,X_j)$ with all of $v$'s neighbors on its boundary. We thus identify the representing cycles in which vertex $v$ can be inserted and connected to its neighbors. We consider each possible choice independently; for each choice we duplicate $C(H,X_j)$ and insert $v$ accordingly. For each of the resulting \compact{s}, we generate all possible shapes. Namely, for each representing cycle, we generate all possible vertex-angle assignments for its vertices and all possible roll-up numbers for its edges that satisfy \cref{le:valid-shape}. Next, for every such assignment, denoted by $\langle \overline{\phi},\overline{\rho} \rangle$, we verify its validity. Let $S_j$ be the set of \shaped{es} $\langle C(H,X_j), \phi, \rho \rangle$ of $B_j$ such that the restriction of $\langle \overline{\phi},\overline{\rho} \rangle$ to the edges of $C(H,X_j)$ corresponds to $\langle \phi, \rho \rangle$.  If  $S_j$ is empty, $\langle \overline{\phi},\overline{\rho} \rangle$ is discarded as it would not be possible to obtain it from $B_j$. Furthermore, observe that $\overline{\rho}(v,u_i)$ corresponds to the number of bends along the edge $(v,u_i)$. Thus, we should ensure that $b^*+\sum_{i=1}^h \overline{\rho}(v,u_i) \le b$, where $b^*$ is the number of bends of  $\langle C(H,X_j), \phi, \rho \rangle$. If this is not the case, again the shape is discarded. Finally, among the putative \shaped{es} generated, we store in $B_i$ only those for which  \cref{le:valid-shape} holds for each of its representing cycles.

\textbf{$X_i$ has two children.} Let $X_j$ and $X_{j'}$ be the children of $X_i$ in $T$. Recall that these three bags are all the same, although $G_j$ and $G_{j'}$ differ. The only orthogonal representations of $G_i$ are those that can be obtained by merging at the common vertices of $X_i$ an orthogonal representation of $G_j$ with an orthogonal representation of $G_{j'}$ in such a way that the resulting representation has a planar embedding, it has at most $b$ bends in total, and it satisfies \cref{def:ortho}. At a high level, this can be done by merging two connected \compact{s} (one in $B_j$ and one in $B_{j'}$) and then by verifying that there is a planar embedding for the merged graph such that \cref{le:valid-shape} is verified for each representing cycle, and the overall number of bends is at most $b$. We split this procedure in two phases.

Let $C(H,X_j)$ be a \compact for which there is at least one \shaped in $B_j$ and let $C(H',X_{j'})$ be a \compact for which there is at least one \shaped in $B_{j'}$. We first compute a connected \compact  $C^*(H,X_j)$ and a connected \compact $C^*(H',X_{j'})$. Let $\overline{C}$ be the union of these two graphs disregarding the rotation system and the choice of the outer face. For each connected component of  $\overline{C}$, we generate all possible planar embeddings. (The embeddings of $\overline{C}$ that are not planar are discarded because they correspond to  non-planar embeddings of $G_i$.) For each planar embedding of $\overline{C}$, we verify that the planar embedding of $\overline{C}$ restricted to the edges of $C(H,X_{j})$ corresponds to the planar embedding of $C(H,X_{j})$ and the same holds for the edges of $C(H',X_{j'})$. This condition ensures that the embedding of $\overline{C}$ can be obtained from those of $C(H,X_j)$ and $C(H',X_{j'})$. We then remove the dummy vertices and the dummy edges from $\overline{C}$ and we analyze each face of the resulting plane graph to verify whether the orientation of its edges is consistent. Namely, a face of a \compact  contains only edges that are either all counterclockwise or clockwise with respect to it. If this condition is not satisfied, the candidate \compact is discarded.

In the second phase, for each  generated \compact, we compute all of its possible shapes and test the validity of each of them. Let $\overline{C}$ be a \compact. For each representing cycle of $\overline{C}$, we generate all possible vertex-angle assignments for its vertices and roll-up numbers for its edges, keeping only those that satisfy \cref{le:valid-shape}. For every such assignment $\langle \overline{\phi},\overline{\rho} \rangle$, let $S_j$ be the set of \shaped{es} $\langle C(H,X_j), \phi, \rho \rangle$ of $B_j$ such that the restriction of $\langle \overline{\phi},\overline{\rho} \rangle$ to the edges of $C(H,X_j)$ corresponds to $\langle \phi, \rho \rangle$. Similarly, let $S_{j'}$ be the set of \shaped{es} $\langle C(H',X_{j'}), \phi', \rho' \rangle$ of $B_{j'}$ such that the restriction of $\langle \overline{\phi},\overline{\rho} \rangle$ to the edges of $C(H',X_{j'})$ corresponds to $\langle \phi', \rho' \rangle$. If any of $S_j$ and $S_{j'}$ is empty, $\langle \overline{\phi},\overline{\rho} \rangle$ is discarded as it would not be possible to obtain it from $B_j$ and $B_{j'}$. Finally, let $b^*_j$ and $b^*_{j'}$ be the minimum number of bends among the \shaped{es} of $S_j$ and $S_{j'}$, respectively. The set $E_i$ of edges shared by $C(H,X_j)$ and $C(H',X_{j'})$  contains edges (if any)  that connect pairs of vertices of $X_i$ and that belong to $G$. In particular, for each edge in $E_i$, the absolute value of its roll-up number corresponds to the number of bends along it. Hence, we verify that $b^*_j+b^*_{j'}-\sum_{(u,v) \in E_i} |\rho(u,v)| \le b$, otherwise we discard $\langle \overline{\phi},\overline{\rho} \rangle$. We conclude:


\begin{lemma}\label{le:correct}
	Graph $G$ admits an orthogonal drawing with at most $b$ bends if and only if algorithm \algo returns a positive answer.
\end{lemma}

\begin{restatable}[*]{lemma}{time}
	\label{le:time}
	Algorithm \algo runs in $k^{O(k)}(b+\sigma)^{k}\log (b+\sigma) \cdot n$ time, where $k$ is the treewidth of $G$, $\sigma$ is the number of degree-two vertices of $G$, and $b$ is the maximum number of bends.
\end{restatable}
\begin{sketch}
	Let $T'$ be a tree-decomposition of $G$ of width $k$ and with $O(n)$ nodes. We compute, in $O(k \cdot n)$ time, a nice tree-decomposition $T$ of $G$ of width $k$ and $O(n)$ nodes~\cite{DBLP:conf/iwpec/BodlaenderBL13,DBLP:books/sp/Kloks94}. In what follows, we prove that \algo spends $k^{O(k)}(b+\sigma)^{k}\log (b+\sigma)$ time for each bag $X_i$ of $T$. The claim trivially follows if $X_i$ is a leaf of $T$.
	If $X_i$ forgets a vertex $v$, \algo considers each \shaped of the child bag $X_j$, which are $k^{O(k)} \cdot (b+\sigma)^{k}$ by \cref{le:num-of-shaped-compact} (a bag of $T$ has at most $k+1$ vertices). For each \shaped, \algo removes $v$ from each of its $O(1)$ representing cycles. Clearly, this can be done in $O(1)$ time. Also, \algo removes possible duplicates in $B_i$. Note that, before removing the duplicates, the elements in $B_i$ are at most as many as those in $B_j$. To efficiently remove the duplicates in $B_i$, we represent each \shaped as a concatenation of three arrays, encoding its \compact (i.e., the rotation system and the outer face), its function $\phi$, and its function $\rho$. Thus we have a set of $N=k^{O(k)}(b+\sigma)^{k}$ arrays each of size $O(k)$. Sorting the elements of this set, and hence deleting all duplicates, takes $O(k) \cdot N \log N$ time, which, with some manipulations, can be rewritten as $k^{O(k)}(b+\sigma)^{k}\log (b+\sigma)$.
	If $X_i$ introduces a vertex $v$,  \algo considers each \compact that can be extracted from $B_j$. Each of them, is then extended with $v$ in all possible ways, which are $k^{O(k)}$ (observe that $|X_j| \le k-1$). Next \algo generates at most $k^{O(k)}(\sigma+b)^{k}$ \shaped{es}. We remark that, as explained in the proof of \cref{le:num-of-shaped-compact}, for each cycle of length $w \le k+1$ it suffices to generate the roll-up numbers of $w-1$ edges. Moreover, for the edges incident to $v$ the roll-up number is restricted to the range $[-b,+b]$ and it is subject to the additional constraint that the total number of bends should not exceed $b$. For each generated shape, \algo checks whether the corresponding subsets of the values of $\phi$ and $\rho$ exist in the \shaped{es} of $B_j$ having the fixed \compact. This can be done by encoding the values of $\phi$ and $\rho$ as two concatenated arrays, each of size $O(k)$, by sorting the set of arrays, and by searching among this set. Since the number of \shaped{es} is $N=k^{O(k)}(\sigma+b)^{k}$, this can be done in $O(k) \cdot N\log N=O(k) \cdot k^{O(k)}(\sigma+b)^{k} \, O(k)\log (k(\sigma+b))=k^{O(k)} (\sigma+b)^k \log (\sigma+b)$ for a fixed \compact, and in $k^{O(k)} (\sigma+b)^{k}\log (\sigma+b)$ time in total. Finally, it remains to check \cref{le:valid-shape} for each representing cycle, which takes $O(k^2)$ time for each of the $k^{O(k)} \cdot (\sigma+b)^{k}$ \shaped{es} in $B_i$.\end{sketch} 

\noindent \cref{le:correct} and \cref{le:time} imply \cref{thm:main} and \cref{co:rpt-bt}.
\section{Applications}\label{se:extensions}

\noindent\textbf{HV Planarity.} Let $G$ be a graph such that each edge is labeled H (horizontal) or V (vertical). \hv asks whether $G$ has a planar drawing such that each edge is drawn as a horizontal or vertical segment according to its label, called a \emph{HV-drawing} (see, e.g.,~\cite{DBLP:journals/jcss/DidimoLP19,DBLP:conf/latin/DurocherF0M14}). This problem is NP-complete~\cite{DBLP:journals/jcss/DidimoLP19}. The next theorem follows~from~our~approach.

\begin{restatable}[*]{theorem}{hvrptbt}\label{th:hvrpt-bt}
	Let $G$ be an $n$-vertex planar graph with $\sigma$ vertices of degree two.
	Given a tree-decomposition of $G$ of width $k$, there is an algorithm that solves \hv in $k^{O(k)}\sigma^{k}\log \sigma \cdot n$ time. 
\end{restatable}

The next corollary  improves the $O(n^4)$ bound in~\cite{DBLP:journals/jcss/DidimoLP19}.

\begin{corollary}
	Let $G$ be an $n$-vertex series-parallel graph.
	There is a $O(n^3 \log n)$-time algorithm that solves \hv.
\end{corollary}

\smallskip

\noindent\textbf{Flexible Drawings.} Let $G=(V,E)$ be a planar graph with vertex degree at most four, and let $\psi: E \rightarrow \mathbb{N}$. The \flex problem~\cite{DBLP:journals/algorithmica/BlasiusKRW14,DBLP:journals/comgeo/BlasiusLR16} asks whether $G$ admits an orthogonal drawing such that for each edge $e \in E$ the number of bends of $e$ is $b(e) \le \psi(e)$. The problem becomes tractable when $\psi(e) \ge 1$~\cite{DBLP:journals/algorithmica/BlasiusKRW14} for all edges, while it can be parameterized by the number of edges $e$ such that $\psi(e)=0$~\cite{DBLP:journals/comgeo/BlasiusLR16}. By subdividing $\psi(e)$ times each edge $e$, we can conclude~the~following.

\begin{restatable}[*]{theorem}{flexdraw}\label{thm:flexdraw}
	Let $G=(V,E)$ be an $n$-vertex planar graph, and let $\psi: E \rightarrow \mathbb{N}$.
	Given a tree-decomposition of $G$ of width $k$, there is an algorithm that solves \flex in $k^{O(k)}(n \cdot b^*)^{k+1}\log(n \cdot b^*)$ time, where $b^* = \max_{e \in E}\psi(e)$.
\end{restatable}

\section{Open Problems}\label{se:conclusions}

The results in this paper suggest some interesting questions. First, we ask whether \ortho is FPT when parameterized by the number of bends and by treewidth. Improving the time complexity of \cref{co:rpt-bt} is also an interesting problem on its own. Since \hv is NP-complete even for graphs with vertex degree at most three~\cite{DBLP:journals/jcss/DidimoLP19}, another research direction is to devise new FPT approaches for  \hv on subcubic planar graphs. 



\clearpage

\bibliographystyle{splncs04}
\bibliography{paper}

\clearpage

\appendix

\section*{Appendix}

\section{Additional Material for \cref{se:compact}}\label{ap:compact}

\numofcompact*
\begin{proof}
	We assume first that $G$ is connected and hence any of its orthogonal representations is connected. 
	
	Let $S(w)$ be the set of all possible sets of representing cycles that can be obtained from $w$ vertices. 
	We show an upper bound on the size of $S(w)$ by showing that each set of representing cycles can be encoded with an array of $O(w)$ cells each storing a value that is $O(w)$.   
	
	In any orthogonal representation of $G$ there are at most $4w$ representing cycles (at most four cycles for each vertex of $X$). Let $v_0,v_1,\dots,v_{w-1}$ be the vertices in $X$ and let $c_1,c_2,\dots,c_{h}$ be a set $R \in S(w)$ of representing cycles of size $h \leq 4w$. For each vertex $v_i$ we use the cells of indices $8i+j$, for $0 \leq j \leq 3$ to represent the cycles which $v_i$ belongs to. Namely, we store in each of these cells a value $1 \leq l \leq h$ to indicate that $v_i$ belongs to $c_l$; if $v$ belongs to $t < 4$ cycles, then $4-t$ cells will have value $0$. Moreover, for each vertex $v_i$ we use the cell of index $8i+j$, for $4 \leq j \leq 7$ to indicate the position of vertex $v_i$ in the cycle stored in the cell of index $8i+j-4$ in clockwise order starting from the vertex with lowest index. Since each set in $S(w)$ can be encoded in this way and since no two sets in $S(w)$ have the same encoding, we have that $|S(w)| \in w^{O(w)}$.
	
	For a given set of representing cycles $R \in S(w)$, let $\mathcal{C}(R)$ be the set of all possible \compact{s} that can be obtained from $R$. We now give a bound on the size of $\mathcal{C}(R)$. This number can be obtained as the product of three terms: the number of distinct \compact{s} that can be built from $R$ disregarding the choice of the rotation system and of the  outer face, the number of distinct rotation systems for each such a graph, and the number of possible choices for the outer face. Observe that for each \compact there is at least one connected \compact and that for each connected \compact there is a unique \compact that produced it, that is, the number of distinct \compact{s} is at most the number of distinct connected \compact{s}. Since there are at most $4w$ dummy vertices in any connected \compact, and since each active vertex is adjacent to at most four dummy vertices by construction, it follows that the first term is $w^{O(w)}$. The second term is also $w^{O(w)}$, because each vertex has degree $O(w)$. The number of faces is $O(w)$ because the graph is planar, it has $O(w)$ vertices, and each edge has multiplicity at most two. It follows that $|\mathcal{C}(R)| \in w^{O(w)}$.   
	
	Thus, we conclude that the number of equivalence classes is 

	$$\sum_{R \in S(w)}{|\mathcal{C}(R)|} \le \sum_{R \in S(w)}{w^{O(w)}} \le |S(w)| \cdot w^{O(w)} \le w^{O(w)} \cdot w^{O(w)} \in w^{O(w)}.$$

	If $G$ is not connected, it has at most $w$ components each containing at least one vertex of $X$, and hence the upper bound derived above should be multiplied by at most $w^w$.
\qed\end{proof}

\validshape*
\begin{proof}
	Let $f$ be the face of $H$ represented by $C_f$. Suppose $f$ is an inner face; the argument is similar for the outer face. Let $n^v_a(f)$ and $n^b_a(f)$ be the number of vertex- and bend-angles, respectively, encountered in a closed walk along the boundary of $f$ whose value is $a \in \{\ph, \p, \pt, 2\p\}$; also, let $n_a(f)=n^v_a(f)+n^b_a(f)$ (note that $n_{2\p}(f)=n^v_{2\p}(f)$ and $n_{\p}(f)=n^v_{\p}(f)$). By definition of roll-up number we have $\rho^*+n_\ph-n_\pt-2n_{2\p}=n_\ph(f)-n_{\pt}(f)-2n_{2\p}(f)$. By \cref{def:ortho} we have $\sum_{(u,v) \in D(f)}(\alpha(u,v)+\beta(v,u)-\beta(u,v))=\pi(\delta(f)-2)$. Moreover, $\sum_{(u,v) \in D(f)}(\alpha(u,v)+\beta(v,u)-\beta(u,v))=\ph \cdot n^v_\ph(f)+ \p \cdot n^v_\p(f) + \pt \cdot n^v_{\pt}(f)+2\p \cdot n^v_{2\p}(f) +\ph \cdot n^b_{\pt}(f) -\ph \cdot n^b_\ph(f)$ and  $\delta(f)=n^v_\ph(f)+n^v_\p(f)+n^v_{\pt}(f)+n^v_{2\p}(f)$. This implies that $\ph \cdot n^v_\ph(f)+ \p \cdot n^v_\p(f) + \pt \cdot n^v_{\pt}(f)+2\p \cdot n^v_{2\p}(f) +\ph \cdot n^b_{\pt}(f) -\ph \cdot n^b_\ph(f) = \p (n^v_\ph(f)+n^v_\p(f)+n^v_{\pt}(f)+n^v_{2\p}(f) - 2)$ and thus $n_\ph(f)-n_{\pt}(f)-2n_{2\p}(f)=4$.
\qed\end{proof}

\numofshapedcompact*
\begin{proof}
	By definition, two orthogonal representations are shape-equivalent if they are $X$-equivalent and their \shaped{es} have the same shape. Thus the number of shape-equivalent classes is given by the number $n_X$ of $X$-equivalent classes multiplied by the number $n_S$ of possible shapes for each $X$-equivalent class. By \cref{le:num-of-compact}, $n_X \leq w^{O(w)}$; we now show that $n_S \leq w^{O(w)}(\sigma+b)^{w-1}$, which proves the statement. For a fixed \compact $C(H,X)$ a shape is defined by assigning to each dart $(u,v)$ of $C(H,X)$ the two values $\phi(u,v)$ and $\rho(u,v,f)$. For each vertex $u$ we have at most four darts exiting from $u$, thus $4w$ in total, each of which can assume one among four values. It follows that the number of choices for the values $\phi(u,v)$ is at most $4^{4w} \leq w^{O(w)}$ (note that $w \geq 2$). As for the possible choices for $\rho(u,v,f)$, each representing cycle $C_f$ of $C(H,X)$ contains at most $w$ vertices and therefore at most $w$ edges. We claim that $-(\sigma+b) \le \rho(u,v,f) \le \sigma+b+4$. The proof of this claim is based on two easy observations. 
	First, the number of vertices and bends forming an angle of $\pt$ inside a face cannot be greater than $b$ plus the number of degree-2 vertices (if a vertex with degree larger than two formed an angle of $\pt$ inside a face, then the other two or three angles around it should sum up to $\ph$). 
	Second, for each vertex forming an angle of $2\p$ inside a face there are two vertex-angles of $\ph$ inside the same face:  
	if a vertex forms an angle of $2\p$, then it has degree one and therefore there must be a cut-vertex forming two $\ph$ 
	angles inside the face (see \cref{fi:2piangle}). 
	Consider a dart $(u,v)$; its roll-up number is defined as $\rho(u,v,f)=n_\ph(u,v)-n_\pt(u,v)-2n_{2\p}(u,v)$. By the second observation above the term $n_\ph(u,v)$ can be written as $2n_{2\p}(u,v)+n'_\ph(u,v)$ (recall that the $2\p$ angles in a face can only be vertex-angles); thus $\rho(u,v,f)=n'_\ph(u,v)-n_\pt(u,v)$. Since $n'_\ph(u,v) \geq 0$ and $n_\pt(u,v) \leq \sigma + b$ (by the first observation above), we have $\rho(u,v,f) \geq - (\sigma+b)$. This proves the lower bound on $\rho(u,v,f)$. Consider now the upper bound. Let $u_1,u_2,\dots,u_h$ be the vertices in $C_f$ and assume that $\rho(u_1,u_2,f) > \sigma+b$. By \cref{le:valid-shape} we have that $\sum_{i=1}^h\rho(u_i,u_{i+1},f)+n_\ph-n_\pt-2n_{2\p} \leq 4$. Each $\rho(u_i,u_{i+1},f)$ can be written as $n_\ph(u_i,u_{i+1})-n_\pt(u_i,u_{i+1})-2n_{2\p}(u_i,u_{i+1})$ and therefore we can write $\rho(u_1,u_{2},f)+\sum_{i=2}^h(n_\ph(u_i,u_{i+1})-n_\pt(u_i,u_{i+1})-2n_{2\p}(u_i,u_{i+1}))+n_\ph-n_\pt-2n_{2\p} \leq 4$. This inequality can be rewritten as $\rho(u_1,u_{2},f)+n'_\ph(f)-n'_\pt(f)-2n'_{2\p}(f) \leq 4$, where $n'_\ph(f)$, $n'_\pt(f)$, and $2n'_{2\p}(f)$ are the number of vertex- and bend-angles of $\ph$, $\pt$, and $2\p$, respectively, inside the face $f$ of $H$ whose representing cycle is $C_f$ in $C(H,X)$, excluded those in the path represented by dart $(u_1,u_2)$. Analogously to the previous case, the contribution of the vertices forming a $2\p$ angle inside $f$ is balanced by a suitable number of vertex-angles of $\ph$. Removing these contributions, we can write the inequality as $\rho(u_1,u_{2},f)+n''_\ph(f)-n'_\pt(f) \leq 4$, that is $\rho(u_1,u_{2},f) \leq -n''_\ph(f)+n'_\pt(f) + 4$. By the first observation above $n'_\pt(f) \leq \sigma+b$ and therefore $\rho(u_1,u_{2},f) \leq \sigma+b + 4$, which proves our upper bound.
	
	\begin{figure}[t]
		\centering
		\begin{minipage}[b]{.25\textwidth}
			\centering 
			\includegraphics[page=1, width=\textwidth]{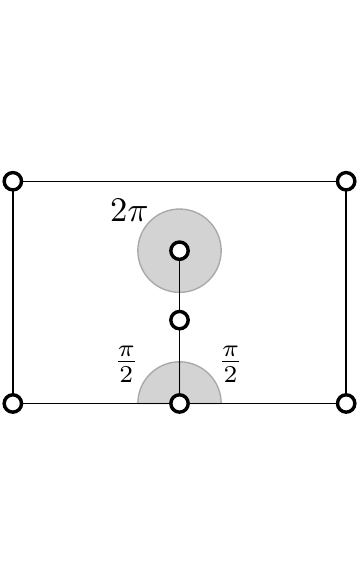}
			\subcaption{~}\label{fi:2piangle}
		\end{minipage}
		\hfil
		\begin{minipage}[b]{.48\textwidth}
			\centering 
			\includegraphics[page=1, width=\textwidth]{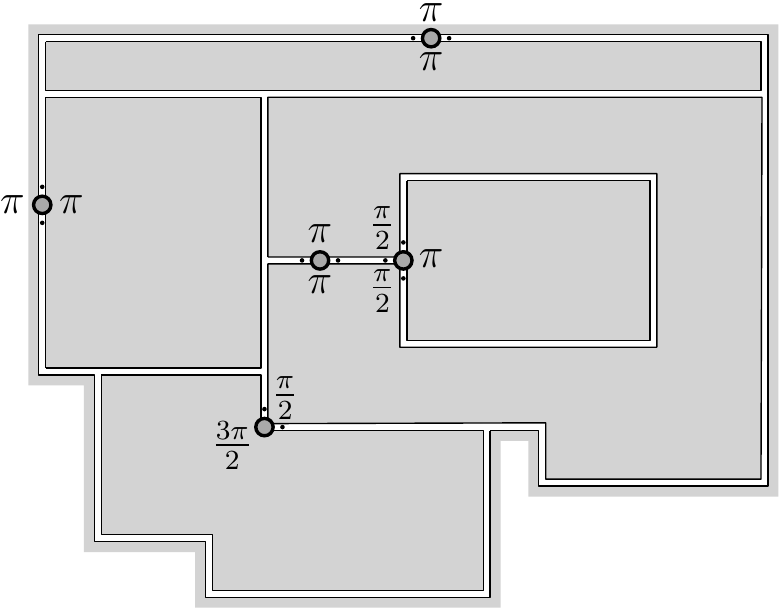}
			\subcaption{~}\label{fi:shape-non-active}
		\end{minipage}
		\caption{(a) For each $2\p$ angle there exist two $\ph$ angles. (b) An orthogonal representation of the largest component of the \compact of \cref{fi:compact-example}; the non active faces (the white ones) have all vertex-angles equal to $0$ (these angles are indicated with a small dot).}
	\end{figure}
	
	In order to conclude the proof, we now show that once the vertex-angles are fixed, the number of darts for which the roll-up number can be fixed independently is at most $w-1$. The key observation is that the roll-up number of the various darts are not independent from one another. For example, if we fix the values of the roll-up number of all the edges of a representing cycle $C_f$ except one (considering that the values of the vertex-angles are also fixed), the roll-up number of the remaining edge is implied because $C_f$ has to satisfy \cref{le:valid-shape}. More in general, the roll-up numbers of the various darts must be assigned in such a way that they correspond to a feasible orthogonal representation. To this aim, we must assign the roll-up numbers in such a way that: (i) the two darts of each edge receive opposite values; (ii) each face of $C(H,X)$ (active or not) satisfies \cref{le:valid-shape}. There is only one difference between active and non-active faces: the vertex-angles of a non-active face are all equal to $0$ (see \cref{fi:shape-non-active} for an example). As a consequence, in the equation of \cref{le:valid-shape} we must take into account the contribution of these vertex-angles, each of which corresponds to two right turns. The equation then becomes $\rho^*+2n_0+n_\ph-n_\pt-2n_{2\p}= \pm 4$ where $n_0$ are the number of vertex-angles equal to $0$. Notice that, this condition must actually be satisfied by any cycle (representing or not) of $C(H,X)$. In particular, for each cycle $\gamma$ in $C(H,X)$ the quantity $\rho(\gamma)=\rho^*+2n_0+n_\ph-n_\pt-2n_{2\p}$ must be equal to $4$ if we consider the darts of the edges oriented so to keep the interior of the cycle to the left, and to $-4$ if we consider the dart of the edges oriented so to keep the exterior of the cycle to the left. We call \emph{inner darts} those in the first set and \emph{outer darts} those in the second set. 
	
	We prove that we can fix independently the roll-up numbers of at most $w-1$ darts of $C(H,X)$. Consider a spanning tree $S$, which has at most $w-1$ edges (it may be a forest). Remove all edges that do not belong to $S$. Observe that, when removing an edge $(u,v)$, we replace two vertex-angles at $u$ with a single vertex-angle at $u$, which is equal to the sum of the original vertex-angles; we do the same for the vertex-angles at $v$. For each edge of $S$ we can choose independently the roll-up number of one of the two darts of that edge, while the second one is implied by the first one. Hence the number of darts that we can fix independently is at most $w-1$. Every edge $(u,v)$ of $C(H,X)$ that is not in $S$ creates a cycle together with the edges of $S$ and therefore if we add the edge $(u,v)$ back (and restore the original values of the vertex-angles), the roll-up number of the dart of $(u,v)$ that is inner for $\gamma$ is implied by the remaining inner darts of $\gamma$, while the roll-up number of dart $(v,u)$ is implied by that of $(u,v)$. Hence, the number of darts that we can fix independently is $w-1$.
	
	
	In conclusion we have $w^{O(w)}$ possible choices for the values $\phi(u,v)$ and for each of these choices we can independently choose the value $\rho(u,v,f)$ for at most $w-1$ darts. We have at most $2(\sigma+b) + 5$ possible values for $\rho(u,v,f)$ (all integer values in the range $[-\sigma-b,\sigma+b+4]$) and hence there are $w^{O(w)}(\sigma+b)^{w-1}$  possible shapes for each $X$-equivalent class, i.e, $n_S \leq w^{O(w)}(\sigma+b)^{w-1}$. 
\qed\end{proof}

\section{Additional Material for \cref{se:algo}}\label{ap:algo}

\time*
\emph{Missing part of the proof.}
We conclude the proof considering the case when $X_i$ joins two bags $X_j$ and $X_{j'}$, \algo first computes a connected \compact of each \compact in $B_j$ and in $B_{j'}$. This takes $k^{O(k)}$ time, as there are $k^{O(k)}$ \compact{s} in $B_j$ and $B_{j'}$ by \cref{le:num-of-compact} and computing a connected \compact takes $O(k)$ time. For each of the $k^{O(k)}$ pairs of \shaped{es} (one in $B_j$ and one in $B_{j'}$),  \algo computes the union of the two graphs and generates all of its possible embeddings, which are $k^{O(k)}$. Consider the union of $\langle C(H,X_j), \phi, \rho \rangle$ and $\langle C(H',X_{j'}), \phi', \rho' \rangle$. For each such unions, \algo checks if the embedding of $C(H,X_j)$ (resp. $C(H',X_{j'})$) is contained in a \compact of $B_j$ (resp. $B_{j'}$). This can be done as follows. Each embedding in $B_j$ (resp. $B_{j'}$) is encoded as an array, the obtained set of arrays is sorted, and the embedding of  $C(H,X_j)$ (resp. $C(H',X_{j'})$) is searched in this set through a binary search. Since there are $N=k^{O(k)}$ \compact{s} in $B_j$ and $B_{j'}$, and since each array encoding an embedding has $O(k)$ size, the sorting takes $O(k) \cdot N \log N=O(k) \cdot k^{O(k)} \, O(k) \log k=k^{O(k)}$ time, and the binary search takes  $O(k) \cdot \log N=O(k) \cdot O(k) \log k =O(k^2\log k)$.

For each union and for each fixed pair of \compact{s}, algorithm \algo generates at most $k^{O(k)}(\sigma+b)^{k}$ \shaped{es}. Observe that, as explained in the proof of \cref{le:num-of-shaped-compact}, for each cycle of length $w \le k+1$ it suffices to generate the roll-up numbers of $w-1$ edges. For each generated shape, \algo checks whether the corresponding subsets of the values of $\phi$ and $\rho$ exist in the \shaped{es} of $B_j$ and $B_{j'}$ having the fixed \compact{s}. This can be done by encoding the values of $\phi$ and $\rho$ as two concatenated arrays, each of size $O(k)$, by sorting the set of arrays, and by searching among this set. Since the number of \shaped{es} is $N=k^{O(k)}(\sigma+b)^{k}$, this can be done in $O(k) \cdot N\log N=O(k) \cdot k^{O(k)}(\sigma+b)^{k} \, O(k)\log (k(\sigma+b))=k^{O(k)} (\sigma+b)^k \log (\sigma+b)$ for a fixed union, and in $k^{O(k)} (\sigma+b)^{k}\log (\sigma+b)$ time in total. Finally, it remains to check \cref{le:valid-shape} for each representing cycle, which takes $O(k^2)$ time for each of the $k^{O(k)} \cdot (\sigma+b)^{k}$ \shaped{es} in $B_i$. 
\qed

	\begin{figure}[htbp]
	\centering
	\begin{minipage}[b]{.4\textwidth}
		\centering
		\includegraphics[page=1, width=\textwidth]{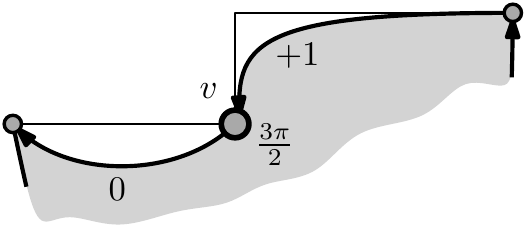}
		\subcaption{~}\label{fi:removesameshape-a}
	\end{minipage}
	\hfil
	\begin{minipage}[b]{.4\textwidth}
		\centering
		\includegraphics[page=2, width=\textwidth]{figures/removesameshape}
		\subcaption{~}\label{fi:removesameshape-b}
	\end{minipage}
	\caption{\label{fi:removesameshape}Two portions of two different \shaped{es} with one and two bends, respectively. When the bigger vertex is removed, they will give rise to two \shaped{es} that are the same but have a different number of bends. }
\end{figure}

\begin{figure}[htbp]
	\centering
	\begin{minipage}[b]{.48\textwidth}
		\centering 
		\includegraphics[page=1, width=\textwidth]{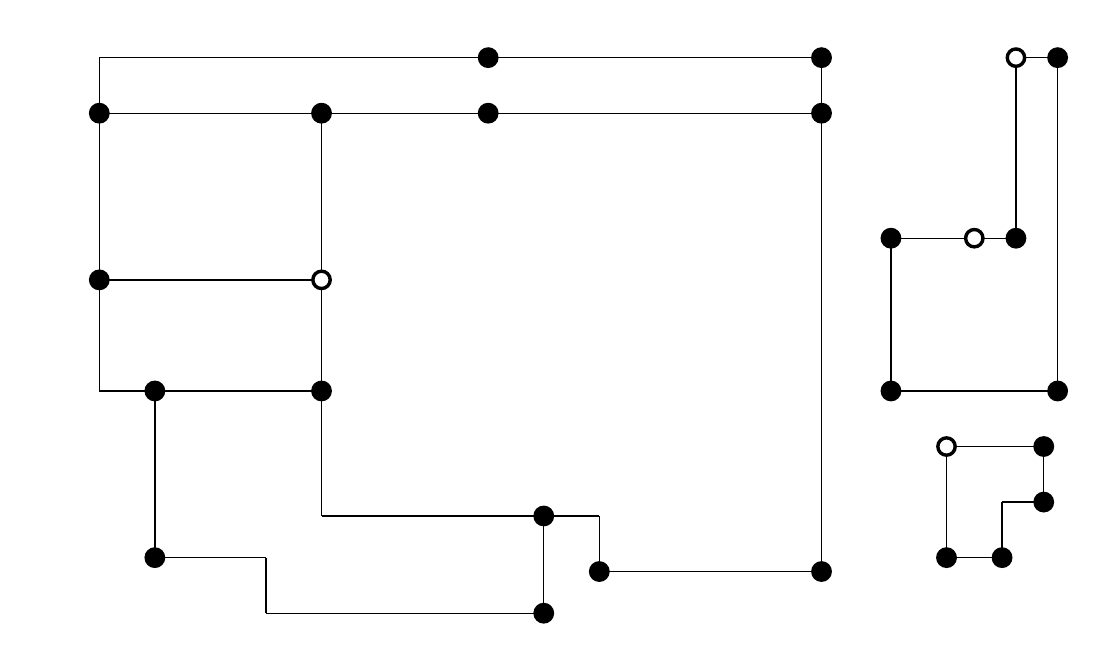}
		\subcaption{~}\label{fi:introduce-disconnected-a}
	\end{minipage}
	\hfil
	\begin{minipage}[b]{.48\textwidth}
		\centering 
		\includegraphics[page=2, width=\textwidth]{figures/introduce-disconnected}
		\subcaption{~}\label{fi:introduce-disconnected-b}
	\end{minipage}
	\begin{minipage}[b]{.48\textwidth}
		\centering 
		\includegraphics[page=3, width=\textwidth]{figures/introduce-disconnected}
		\subcaption{~}\label{fi:introduce-disconnected-c}
	\end{minipage}
	\hfil
	\begin{minipage}[b]{.48\textwidth}
		\centering 
		\includegraphics[page=4, width=\textwidth]{figures/introduce-disconnected}
		\subcaption{~}\label{fi:introduce-disconnected-d}
	\end{minipage}
	\caption{\label{fi:introduce-disconnected} Extending a disconnected orthogonal representation with an introduced vertex $v$. Vertex $v$ and its neighbors are shown in white. (a)--(b) There exists an inner face containing a neighbor of $v$;  (c)--(d) No internal face contains a neighbor of $v$. }
\end{figure}

\begin{figure}[htbp]
	\centering
	\begin{minipage}[b]{.48\textwidth}
		\centering 
		\includegraphics[page=1, width=\textwidth]{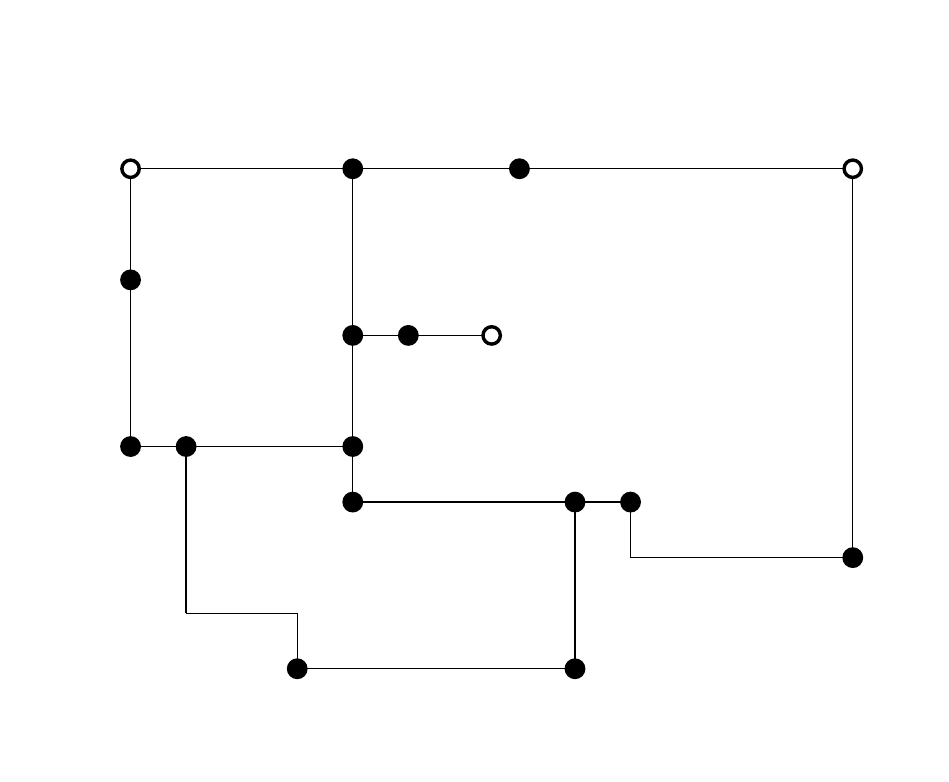}
		\subcaption{~}\label{fi:join-a}
	\end{minipage}
	\hfil
	\begin{minipage}[b]{.48\textwidth}
		\centering 
		\includegraphics[page=3, width=\textwidth]{figures/join}
		\subcaption{~}\label{fi:join-b}
	\end{minipage}
	\begin{minipage}[b]{.48\textwidth}
		\centering 
		\includegraphics[page=4, width=\textwidth]{figures/join}
		\subcaption{~}\label{fi:join-c}
	\end{minipage}
	\hfil
	\begin{minipage}[b]{.48\textwidth}
		\centering 
		\includegraphics[page=6, width=\textwidth]{figures/join}
		\subcaption{~}\label{fi:join-d}
	\end{minipage}
	\begin{minipage}[b]{.48\textwidth}
		\centering 
		\includegraphics[page=8, width=\textwidth]{figures/join}
		\subcaption{~}\label{fi:join-e}
	\end{minipage}
	\hfil
	\begin{minipage}[b]{.48\textwidth}
		\centering 
		\includegraphics[page=7, width=\textwidth]{figures/join}
		\subcaption{~}\label{fi:join-f}
	\end{minipage}
	\caption{\label{fi:join} Merging two orthogonal representations sharing a cutset $X$. Vertices in $X$ are shown in white. (a)--(b) A rectilinear representation $H_1$ and its connected \compact $C^*(H_1,X)$ with the shape corresponding to $H_1$. (c)--(d) A rectilinear representation $H_2$ and its connected \compact $C^*(H_2,X)$ with the shape corresponding to $H_2$.  (e)--(f) The union of the two connected \compact{s} and a corresponding rectilinear representation.}
\end{figure}

\section{Additional Material for \cref{se:extensions}}\label{ap:extensions}

\hvrptbt*
\begin{proof}
	We modify \algo as follows. First of all, since in an HV-drawing the edges have no bends, we set $b=0$. Second, in order to ensure that an edge is drawn according to its label (H or V), we keep a \emph{reference direction} for each edge of each representing cycle in an \shaped. Observe that an edge $e=(u,v)$ is added to an \shaped when $u$ is in a bag $X_j$ and $v$ is introduced by the parent bag $X_i$ of $X_j$. Recall that an edge of a representing cycle corresponds to a path (possibly consisting of a single edge). Its reference direction is either H or V based on whether the first edge of the path is horizontal or vertical in every orthogonal representation encoded by that \shaped.  This information,  together with the roll-up number of that edge and with the vertex-angle $\phi(u,v)$ (or $\phi(v,u)$ based on the orientation of the edge) suffice to determine the direction of $e$ and hence to verify whether its direction is consistent with its label. Clearly, if $u$ is an isolated vertex, there is no reference direction and $e$ can be drawn with the desired direction.
\qed\end{proof}

\flexdraw*
\begin{proof}
	It suffices to subdivide $\psi(e)$ times each edge $e$ of $G$ and then decide whether the resulting graph $G'$ admits an orthogonal drawing with $b=0$. If such a drawing exists, each subdivision vertex can be interpreted either as a bend if one of the two angles it makes is greater than $\p$, or as an interior point of a segment otherwise. On the other hand, if an orthogonal drawing of the original graph $G$ exists such that each edge $e$ has at most $\psi(e)$ bends, then each bend of an edge $e$ can be replaced with a subdivision vertex. Also, if the number of bends $b(e)$ of $e$ is less than $\psi(e)$, we can further subdivide $b(e)-\psi(e)$ times any segment representing $e$ so to obtain an orthogonal drawing of $G'$ with $b=0$.    
\qed\end{proof}
\end{document}